\documentclass[letterpaper,10pt]{article}
\usepackage{times}
\usepackage{helvet}
\usepackage{courier}
\usepackage{amsfonts}   % if you want the fonts
\usepackage{amsmath,amssymb,amstext,comment} % Lots of math symbols and environments
\usepackage{amsthm}
\usepackage{graphicx}
\usepackage{epstopdf}
\usepackage{enumerate}
\usepackage[linesnumbered,ruled,vlined,boxed]{algorithm2e}
\usepackage{url}
\usepackage{hyperref}
\newtheorem{definition}{Definition}
\newtheorem{proposition}{Proposition}

\newtheorem{corollary}{Corollary}

\newtheorem{theorem}{Theorem}
\newtheorem{lemma}{Lemma}

\newcommand{\footnoteremember}[2]{
\footnote{#2}
\newcounter{#1}
\setcounter{#1}{\value{footnote}}
}
\newcommand{\footnoterecall}[1]{
\footnotemark[\value{#1}]
}

\begin{document}

%\frenchspacing
%%%\pdfinfo{
%%%/Title (Formatting Instructions for Authors Using LaTeX)
%%%/Subject (AAAI Publications)
%%%/Author (AAAI Press)}
%%%\setcounter{secnumdepth}{0}
\title{Equilibria of Chinese Auctions}
\author{
Simina Br\^{a}nzei
\footnoteremember{AU}{Department of Computer Science, Aarhus University, Denmark. Email: \{simina,bromille\}@cs.au.dk}
\and
Clara Forero
\footnoteremember{UW}{School of Computer Science, University of Waterloo, Canada. Email: \{ciforero, klarson\}@uwaterloo.ca}
\and
Kate Larson
\footnoterecall{UW}
\and
Peter Bro Miltersen
%\footnote{Aarhus University}
\footnoterecall{AU}
}
\date{}

%\author{Simina Br\^{a}nzei, Clara Forero, Kate Larson, and Peter Bro Miltersen}

\maketitle

%\pagestyle{plain}

%\maketitle
% The file aaai.sty is the style file for AAAI Press
% proceedings, working notes, and technical reports.
%
%\numberofauthors{2}

% You can go ahead and credit any number of authors here,
% e.g. one 'row of three' or two rows (consisting of one row of three
% and a second row of one, two or three).
%
% The command \alignauthor (no curly braces needed) should
% precede each author name, affiliation/snail-mail address and
% e-mail address. Additionally, tag each line of
% affiliation/address with \affaddr, and tag the
% e-mail address with \email.
% 1st. author
%}
%\begin{comment}
%\author{
%Simina Br\^{a}nzei\\
%Aarhus University, Denmark\\
%Association for the Advancement of Artificial Intelligence\\
%445 Burgess Drive\\
%Menlo Park, California 94025\\
%simina@cs.au.dk\\
%\And
%Clara Forero\\
%\And
%Peter Bro Miltersen\\
%Aarhus University, Denmark\\
%bromille@cs.au.dk\\
%}
%\end{comment}

\begin{abstract}
Chinese auctions are a combination between a raffle and an auction and are held in practice at charity events or festivals. In a Chinese auction, multiple players compete for several items by buying tickets, which can be used to win the items. In front of each item there is a basket, and the players can bid by placing tickets in the basket(s) corresponding to the item(s) they are trying to win. After all the players have placed their tickets, a ticket is drawn at random from each basket and the item is given to the owner of the winning ticket. While a player is never guaranteed to win an item, they can improve their chances of getting it by increasing the number of tickets for that item.
  In this paper we investigate the existence of pure Nash equilibria in both the continuous and discrete settings. When the players have continuous budgets, we show that a pure Nash equilibrium may not exist for asymmetric games when some valuations are zero. In that case we prove that the auctioneer can stabilize the game by placing his own ticket in each basket. On the other hand, when all the valuations are strictly positive, a pure Nash equilibrium is guaranteed to exist, and the equilibrium strategies are symmetric when both valuations and budgets are symmetric. We also study Chinese auctions with discrete budgets, for which we give both existence results and counterexamples. While the literature on rent-seeking contests traditionally focuses on continuous costly tickets, the discrete variant is very natural and more closely models the version of the auction held in practice.
\end{abstract}

%\maketitle
\section{Introduction}
Chinese auctions are a combination between a raffle and an auction and are held in practice at charity events or festivals \cite{wiki}.
In a Chinese auction, multiple players compete for several items by buying tickets, which can be used to win the items.
In front of each item there is a basket, and the players can bid by placing tickets in the basket(s) corresponding to the item(s) they are trying to win.
After all the players have placed their tickets, a ticket is drawn at random from each basket and the item is given to the owner of the winning ticket.
While a player is never guaranteed to win an item, they can improve their chances of getting it by increasing the number
of tickets for that item.

Chinese auctions are related to the rent seeking contest introduced by Tullock in 1980 \cite{Tullock80}. 
In the Tullock contest, two players compete for winning a prize. The probability that a player wins the prize 
is a function of both players' effort. A player is never guaranteed to receive the item,
but can increase his chances of winning
it by increasing his effort. There exist two main types of rent seeking contests, namely perfectly discriminating
and imperfectly discriminating. In a perfectly discriminating contest, having the highest amount of
effort secures a win. Perfectly discriminating contests are a generalization of all-pay auctions and have been
studied, for example, in Moldovanu and Sela \cite{Moldovanu01}. The authors analyze the optimal allocation of multiple
prizes with symmetric players and prove the existence of symmetric bidding equilibria for contestants with
linear, convex, and concave cost functions. In an imperfectly discriminating contest a player is never guaranteed
to get an item, unless he is the only one exerting effort to obtain it. Both Chinese auctions and Tullock's
original model are imperfectly discriminating contests. However, Chinese auctions are a generalization of
the Tullock contest when the exponent is R = 1 and multiple players compete for multiple prizes, where the
players have asymmetric valuations. We study auctions with both costly and given budgets - the latter variant
is traditionally not analyzed in the rent seeking literature, but is a natural model related to several classes of
games such as threshold task games \cite{Chalkiadakis10} or coalitional skill games \cite{Bachrach03}. In addition, the items that no player
placed a bid on are kept by an auctioneer, while in rent seeking literature it is commonly assumed that such
items are assigned to some player at random.

Gradstein and Nitzan \cite{Gradstein89} study a generalization of the rent seeking contest where the players 
are identical and the number of items is restricted in a certain range, characterize the pure Nash equilibria of the game, and give conditions
for the existence of mixed Nash equilibria when $n$ is large.
Nitzan \cite{Nitzan94} surveys rent seeking contests and describes settings where multiple players compete for rent under different 
assumptions regarding the number of players, their risk attitudes, the source and nature of the rent.
Chowdhury and Sheremeta \cite{Chowdhury11} study another generalization of the Tullock contest, in which two players compete for
two prizes. The probabilities of winning are defined as in the Tullock model, but the payoff of each player (contingent upon winning or losing)
is a linear function of the prizes, own effort, and effort of the rival.
Nti \cite{Nti99} studies the Tullock contest with asymmetric valuations of the two players for the prize and variable ranges of the return to
scale parameter, and establishes a necessary and sufficient condition for the existence of a unique pure Nash equilibrium.
Hillman and Riley \cite{Hillman89} study a lottery model with multiple players and one prize and prove the existence of an equilibrium
for both discriminating and indiscriminating contests, symmetric and asymmetric valuations. Fang \cite{Fang02} shows that the equilibrium identified
by Hillman and Riley is in fact unique.
Siegel \cite{Siegel09} studies perfectly discriminating contests with multiple players and multiple identical prizes, provides a closed form
solution for the equilibrium payoffs, and analyzes player participation.

The model closest to ours in a published paper is by Palma and Munshi \cite{Palma12}, where multiple 
players compete for multiple prizes in an imperfectly discriminating contest.
The paper focuses on defining a 'holistic' probability model, in
which the effort of the players are mapped to the aggregate probability of a possible outcome. From the model,
they derive the probability of a player being successful, and show that when the costs of effort are symmetric
among players, then a symmetric Nash equilibrium exists. 
Matros \cite{Matros07} considers the exact same model of Chinese auctions as us and states the existence
of a symmetric pure Nash equilibrium when the valuations are symmetric and the existence of a pure Nash equilibrium 
with asymmetric valuations, for both costly and given tickets.
However, we are not aware of the existence of any full paper with proofs of the stated results.
On the contrary to the abstract, we show that a pure Nash equilibrium may not exist for asymmetric games when some valuations 
are zero.
In that case we prove that the auctioneer can stabilize the game by placing his own
ticket in each basket. On the other hand, when all the valuations are strictly positive, a pure Nash equilibrium
is guaranteed to exist, and the equilibrium strategies are symmetric when both valuations and budgets are
symmetric. We also study Chinese auctions with discrete budgets, for which we give both existence results
and counterexamples. While the literature on rent-seeking contests traditionally focuses on continuous costly
tickets, the discrete variant is very natural and more closely models the version of the auction held in practice.

\section{The Model}
Let $N = \{1, \ldots, n\}$ be a set of players and $M = \{1, \ldots, m\}$ a set of prizes.
Each player has several lottery tickets which are chances to win the items.
The players bid by placing tickets in a basket in front of each item they are trying to win.
After all the players distribute their tickets, a lottery is held at each item.
One ticket is drawn at random from the basket of the item, and the item is given to the owner 
of that ticket.
The items that no player placed a bid on are kept by the auctioneer.

Formally, for each player $i$, let $w_i$ be the total weight of $i$'s tickets. For each item
$j$, let $v_{i,j}$ denote the valuation of player $i$ for item $j$.
We study several types of budgets:
\begin{itemize}
\item \textit{Discrete budget}: Each player has several indivisible tickets of weights 
$T_i = \{t_{i,1}, \ldots, t_{i,n_i}\}$, where $\sum_{j = 1}^{n_i} t_{i,j} = w_i$.
The player can distribute the tickets as he wishes across the bins.
\item \textit{Continuous budget}: Each player $i$ has a budget $w_i$. The player can distribute $w_i$ 
arbitrarily across the items. If $w_{i,j}$ is the weight placed by player $i$ on each item $j$, it 
must be the case that $\sum_{j=1}^{m} w_{i,j} = w_i$ and $w_{i,j} \geq 0$, $\forall j \in M$.
\end{itemize}

\begin{comment}
\item \textit{Indivisible budget}: Each player $i$ has an indivisible ticket of weight $w_i \in \mathbb{R}^{+}$. 
The player can place the ticket in any bin of his choice.
\item \textit{Coin-based budget}: Each player has several indivisible tickets of weights $T_i = \{t_{i,1}, \ldots, t_{i,n_i}\}$,
where $\sum_{j = 1}^{n_i} t_{i,j} = w_i$.
The player can distribute them as he wishes across the bins.
\end{comment}

We first study the setting where the players are endowed with the tickets and pay no cost for obtaining them.
However, the budget of each player $i$ is limited and possibly different from that of the other players.
\begin{definition} \label{def:utility}
Given assignment $x = \left( x_{i,j} \right)_{i \in N, j \in M}$ of tickets to items, where $x_{i,j}$ is the 
weight of the tickets placed by player $i$ on item $j$, the expected utility of player $i$ is:
\begin{equation} \label{def:utility_given}
u_i(x) = \sum_{j = 1}^{m} \sigma_{i,j}(x) v_{i,j}
\end{equation}
where
\[
\sigma_{i,j}(x) =
\left\{
  \begin{array}{ll}
    \frac{x_{i,j}}{\sum_{k=1}^{n} x_{k,j}} & \mbox{if} \sum_{k=1}^{n} x_{k,j} > 0\\
    0 & \mbox{otherwise}
  \end{array}
\right.
\]
\end{definition}

An assignment $x$ of tickets to items is a pure Nash equilibrium if for every player $i \in N$ and any other assignment $y_i$ of tickets to items
by player $i$, the following holds: $u_i(y_i, x_{-i}) \leq u_i(x)$.

The rest of the paper is organized as follows. In Section \ref{sec:given} we study the model introduced in Definition \ref{def:utility},
for both continuous and discrete budgets. In the case of continuous budgets (Section \ref{sec:given_costly}), we show that a pure Nash equilibrium is guaranteed
to exist when all the valuations are strictly positive, and provide a closed form solution of equilibrium strategies for symmetric valuations. 
When the valuations can be zero, then a pure Nash equilibrium may not exist. In that case, the auctioneer can ensure existence
by placing his own ticket in each basket, such that no player gets the item if the auctioneer ticket is drawn.
For discrete budgets (Section \ref{sec:costly}, we give both existence results and counterexamples, depending on whether the players
have more than one ticket.
Finally, in Section \ref{sec:costly_continuous}, we study costly continuous budgets. With costly tickets,
a pure Nash equilibrium is guaranteed to exist when the valuations are strictly positive. When the valuations can be zero, then an equilibrium
may fail to exist, and similarly to the continuous given tickets scenario, the auctioneer can help by placing his own ticket in each basket.

%%%%%%%%%%%%%%%%%%%%%%%%%%%%%%%%%%%%%%%%%%% Format of the paper %%%%%%%%%%%%%%%%%%%%%%%%%%%%%%%%%%%%%%%%%

\section{Given Tickets} \label{sec:given}
In this section we study the game where the players are endowed with a budget of tickets, which they can use to maximize their chances
to win the items.

\subsection{Continuous Budgets} \label{sec:given_costly}
We use the following lemma.
\begin{lemma} \label{lem:convex}
Let $f : S^{m-1} \rightarrow \mathbb{R}$, where 
$S^{m-1} = \{ y \in \mathbb{R}^{m} | y_i \geq 0 \; \mbox{and} \; y_1 + \ldots + y_m = W\}$. Define 
% be such that for all $y = (y_1, \ldots, y_m) \in \Delta^{m-1}$, 
$f(y) = \sum_{j=1}^{m} \frac{b_j}{a_j + y_j}$, where $a_j, b_j > 0, \forall j \in \{1, \ldots, m\}$.
Then $f$ is strictly convex.
\end{lemma}
%%%%%%%%%%%%%%%%%%Proof in the appendix%%%%%%%%%%%%%%%%%%%
%\begin{proof}
%The domain of $f$ is convex, thus it is sufficient to verify that for all $y \neq z \in S^{m-1}$ and $\lambda \in (0,1)$, $f(\lambda y + (1 - \lambda) z) < \lambda f(y) + (1-\lambda) f(z)$.
%For each $j \in \{1, \ldots, m\}$, let $f_j: \mathbb{R}^{+} \leftarrow \mathbb{R}$, $f_j(x) = \frac{b_j}{a_j + y_j}$. Then $f_{j}^{''} = \frac{2 b_j}{(a_j + y_j)^{3}} > 0$, and $f_j$ is strictly convex. Thus
%$\frac{a_j \cdot b_j}{a_j + \lambda y_j + (1 - \lambda) z_j} \leq \lambda \left(\frac{a_j \cdot b_j}{a_j + y_j}\right) + (1 - \lambda) \left(\frac{a_j \cdot b_j}{a_j + z_j}\right)$ $(^*)$
%with equality if and only if $y_j = z_j$. By summing $(^*)$ over all $j$ and noting that at least one inequality is strict (since $y \neq z$), we obtain that $f$ is strictly convex.
%\end{proof}

\begin{theorem} \label{thm:symmetric_cont_val}
Chinese auctions with symmetric valuations and continuous budgets have a pure Nash equilibrium in which
all the players allocate the same percentage of their budget on a given item.
\end{theorem}
\begin{proof}
Since the valuations are symmetric, we can assume without loss of generality that all the values are strictly positive.
We show that the allocation in which every player $i$ allocates amount $x_j = \left(\frac{v_j}{v_1 + \ldots + v_m}\right)w_i$
on item $j$ is a pure Nash equilibrium.

If all the other $n-1$ players allocate $x_j$ (as defined above) on every $s_j$,
the utility of player $i$ when using allocation $y= (y_1, \ldots, y_m)$ is:
\[
u_i(y, x_{-i}) = \sum_{j=1}^{m} \frac{y_j \cdot v_j}{\left( \sum_{k \neq i} \frac{(W - w_i)v_j}{v_1 + \ldots + v_m} \right) + y_j}
\]
where $y_j \geq 0$, $\sum_{j=1}^{m} y_j = w_i$, and $W = \sum_{l = 1}^{n} w_l$.
We claim that when the other players allocate $x_{-i}$, the best response of player $i$ is to allocate $x_i = (x_{i,1}, \ldots, x_{i,m})$.
Player $i$'s utility can be rewritten as:
\[
u_i(y, x_{-i}) = \sum_{j=1}^{m} v_j - \sum_{j=1}^{m} \frac{\left( \frac{(W - w_i)v_j^2}{v_1 + \ldots + v_m} \right)}{\left( \frac{(W - w_i)v_j}{v_1 + \ldots + v_m} \right) + y_j}
\]
Let $b_{i,j} = \frac{(W - w_i)v_j^2}{v_1 + \ldots + v_m}$ and $a_{i,j} = \frac{(W - w_i)v_j}{v_1 + \ldots + v_m}$, $\forall j \in M$. Define
$f_i(y) : S_{i}^{m-1} \rightarrow \mathbb{R}$, where $S_{i}^{m-1} = \{ y \in \mathbb{R}^{m} | y_j \geq 0, \forall j \in \{1, \ldots, m\} \; \mbox{and} \; y_1 + \ldots + y_m = w_i \}$
by $f_i(y) = \sum_{j=1}^{m} \frac{b_j}{a_j + y_j}$.
An allocation $y$ maximizes player $i$'s utility if and only if $f_i$ has a global minimum at $y$.
By Lemma \ref{lem:convex}, $f_i$ is strictly convex. Then $f_i$ has a unique global minimum, and moreover any 
local minimum is also a global minimum.

Let $g_{i,j}(y) = - y_j$, $\forall j \in \{1, \ldots, m\}$ and $h_i(y) = y_1 + \ldots + y_m - w_i$.
Finding the global minimum of $f_i$ is equivalent to solving the following optimization problem:
\begin{eqnarray*}
\mbox{min} & f_i(y) \\
\mbox{s.t.} \; & g_{i,j}(y) \leq 0, \forall j \in \{1, \ldots, m\} \\
& h_i(y) = 0
\end{eqnarray*}
Since $f_i$ and $g_{i,1}, \ldots, g_{i, m}$ are continuously differentiable convex functions and $h_i$ is an affine function, the $KKT$
conditions are both necessary and sufficient for a point $y$ to be a global minimum of $f_i$.
Let $\mu_{i,j} = 0, \forall j \in \{1, \ldots, m\}$ and $\lambda_i = \frac{(W - w_i)(v_1 + \ldots + v_m)}{W^2}$.
The KKT conditions at $x$ are:
\begin{enumerate}
\item $\nabla f_i(x) + \sum_{j=1}^{m} \mu_{i,j} \nabla g_{i,j}(x) + \lambda \nabla h_i(x) = 0$: That is,
$\begin{pmatrix} \frac{-b_{i,1}}{(a_{i,1} + x_{i,1})^2} \\ \ldots \\ \frac{-b_{i,m}}{(a_{i,m} + x_{i,m})^2} \end{pmatrix} + 
\begin{pmatrix} \lambda_i \\ \ldots \\ \lambda_i \end{pmatrix} = 0$, or $\lambda_i = \frac{b_{i,j}}{(a_{i,j} + x_{i,j})^2}$, $\forall j \in \{1, \ldots, m\}$,
which is immediate from the definitions of $a_{i,j}, b_{i,j}$, and $x_{i,j}$.
\item $g_{i,j}(x) \leq 0$, $\forall j \in \{1, \ldots, m\}$ and $h_i(x) = 0$: That is, $x_{i,j} \geq 0$, $\forall j \in \{1, \ldots, m\}$ 
and $x_{i,1} + \ldots + x_{i,m} = w_i$, which follows from the definition of $x$.
\item $\mu_{i,j} \geq 0$, $\forall j \in \{1, \ldots, m\}$: By definition, $\mu_{i,j} = 0$, $\forall j \in \{1, \ldots, m\}$.
\item $\mu_{i,j} g_{i,j}(x) = 0$, $\forall j \in \{1, \ldots, m\}$: Immediate since $\mu_{i,j} = 0$, $\forall j \in \{1, \ldots, m\}$.
\end{enumerate}
Thus the best response of player $i$ when all the other players allocate $x_{-i}$ is to also allocate according to $x$,
and so the game has a pure Nash equilibrium in which every player $i$ allocates $x_{i,j} = \left(\frac{v_j}{v_1 + \ldots + v_m}\right)w_i$ on every
item $j$.
\end{proof}

We obtain the following corollary when both the valuations and budgets are symmetric.

\begin{corollary}
Chinese auctions with symmetric valuations and symmetric continuous budgets have a symmetric pure Nash equilibrium.
\end{corollary}
\begin{proof}
Without loss of generality, we can assume that all the budgets are $1$.
By applying Theorem \ref{thm:symmetric_cont_val}, we obtain that the allocation $x_{i,j} = \frac{v_j}{v_1 + \ldots + v_m}$ for every player $i$ and item
$j$ is a pure Nash equilibrium.
\end{proof}

\begin{proposition}
Chinese auctions with asymmetric valuations and asymmetric continuous budgets do not necessarily have a pure Nash equilibrium when there exist zero valuations.
\end{proposition}
\begin{proof}
Consider a two player game with two items with the following valuations and budgets: $v_{1,1} = 0, v_{1,2}=1, v_{2,1}=1, v_{2,2} = 3$
and $w_1 = w_2 = 1$. Assume by contradiction that the game has a pure Nash equilibrium, $x^* = (x_1^{*}, x_{2}^{*}) \in [0,1]$.
First note that $x_{1}^{*} = 0$, since otherwise player $1$ can improve his utility by allocating more weight on item $2$.
For all $\varepsilon > 0$, $x_{2}^{*} = 1 - \varepsilon$ cannot be a Nash equilibrium, since $u_2(1, 1 - \varepsilon) < u_2(1, 1 - \frac{\varepsilon}{2})$.
Thus the only remaining candidate for an equilibrium is $x^* = (1, 1)$. However, there exists $\varepsilon > 0$ such that $u_2(1, 1 - \varepsilon) > u_2(1,1)$.
Thus the game has no pure Nash equilibrium.
\end{proof}

%\textbf{TODO\#:} See what happens if the auctioneer puts a ticket of value $\Delta_j$ in each basket $s_j$, where $\Delta_j %<< w_i$ for all $i$. I think this should ensure existence of the equilibrium.
%\textbf{TODO\#:} Consider $\varepsilon$ equilibrium, I think this should exist.

When the players may have zero valuations, the auctioneer can guarantee the existence of a pure strategy equilibrium by 
placing a small ticket in each basket, such that if the auctioneer ticket is drawn, no player gets the item.
To prove this, we use the following theorem \cite{Reny06}.

\begin{theorem} (Debreu 1952; Glicksberg 1952; Fan 1952) \label{thm:Glicksberg}
Consider a strategic form game whose strategy spaces $S_i$ are nonempty compact convex subsets of an Euclidean space. If the payoff functions $u_i$ are continuous in $s$ and quasi-concave in $s_i$, then there exists a pure strategy Nash equilibrium.
\end{theorem}

\begin{comment}
We first introduce a lemma.

\begin{lemma} \label{lem:quasi_concave}
Let $f : S^{m-1} \rightarrow \mathbb{R}$, where 
$S^{m-1} = \{ y \in \mathbb{R}^{m} | y_i \geq 0 \; \mbox{and} \; y_1 + \ldots + y_m = W\}$.
For each $j \in \{1, \ldots, m\}$, let $a_j, b_j \geq 0$ and define
$f_j : \mathbb{R}^{+} \rightarrow \mathbb{R}$, 
$f_j(x) =  \frac{b_j \cdot x}{a_j + x}$ if $x > 0$ and $f_j(x) = 0$ otherwise.
Let $f(y) = \sum_{j=1}^{m} f_j(y_j)$. Then $f$ is quasi-concave.
\end{lemma}
\end{comment}

%We first note that the utility function of each player $i$, $u_i(x)$, is strictly concave in their own strategy, $x_i$.

\begin{theorem} \label{thm:auctioneer_ticket}
Chinese auctions with asymmetric valuations and asymmetric continuous budgets have a pure Nash equilibrium when the auctioneer places a strictly positive ticket in each basket.
\end{theorem}
\begin{proof}
Let $\Delta_j > 0$ be the ticket placed by the auctioneer in each basket. Let $x_{i,j}$ be the weight placed by player $i$ on item $j$, where $x_{i,j} \geq 0$ and $\sum_{j=1}^{m} x_{i,j} = w_i$ for all $i \in N$. The utility of player $i$ is:
\begin{equation} \label{eq:utility_delta_ticket}
u_i(x) = \sum_{j=1}^{m} \left( \frac{x_{i,j}}{\Delta_j + X_{j}^{-i} + x_{i,j}} \right) v_{i,j}
\end{equation}
where $X_{j}^{-i} = \sum_{k \neq i} x_{k,j}$ is the weight placed by all players except $i$ on item $j$.
%By Lemma \ref{lem:quasi_concave}, $u_i(x)$ is quasi-concave in $x_i$. 
It can be easily verified that the utility function of each player $i$, $u_i(x)$, is strictly concave in their own strategy, $x_i$.
The strategy spaces $S_i = \{y \in \mathbb{R}^{m} | y_j \geq 0, y_1 + \ldots + y_m = w_i\}$ are nonempty, compact, and convex. Moreover, $u_i(x)$ is continuous in $x$ since the denominator of each term in the sum of Equation (\ref{eq:utility_delta_ticket}) is strictly positive.
Thus the conditions of Theorem \ref{thm:Glicksberg} apply, and the game has a pure strategy Nash equilibrium when the auctioneer places a ticket in each basket.
\end{proof}

When all the valuations are strictly positive, a pure Nash equilibrium is guaranteed to exist. To prove this, we use the following result by Reny \cite{Reny06}
for discontinuous games. First, we define the \emph{better-reply secure} property of a game.

\begin{definition}
Player $i$ can \emph{secure} a payoff of $\alpha \in \mathbb{R}$ at $s \in S$ if there exists $\bar{s_i} \in S_i$, such that $u_i(\bar{s_i}, s_{-i}^{'}) \geq \alpha$ for all $s_{-i}^{'}$ close enough to $s_{-i}$.
\end{definition}

\begin{definition}
A game $G = (S_i, u_i)_{i=1}^{n}$ is \emph{better-reply secure} if whenever $(s^{*}, u^{*})$ is in the closure of the graph of its vector payoff function and $s^{*}$ is not a Nash equilibrium, then some player $i$ can secure a payoff strictly above $u_i^{*}$ at $s^{*}$.
\end{definition}

\begin{theorem}[Reny, 1999] \label{thm:Reny}
If each $S_i$ is a nonempty, compact, convex subset of a metric space, and each $u_i(s_1, \ldots, s_n)$ is quasi-concave in $s_i$, then the game $G = (S_i, u_i)_{i=1}^{n}$ has at least one pure Nash equilibrium if in addition $G$ is better-reply secure.
\end{theorem}

We now prove the main result of this section.
\begin{theorem} \label{thm:asymmetric_continuous_existence}
Chinese auctions with asymmetric, strictly positive valuations and asymmetric continuous budgets have a pure Nash equilibrium.
\end{theorem}
\begin{proof}
The strategy spaces $S_i = \{y \in \mathbb{R}^{m} | y_j \geq 0, y_1 + \ldots + y_m = w_i\}$ are nonempty, compact, and convex.
%By Lemma \ref{lem:quasi_concave}, 
Again, the utility function of each player $i$ is strictly concave in $x_i$ (and thus it is also quasi-concave). We show that the game is also better-reply secure.
By Reny \cite{Reny06}, all games with continuous payoffs are better-reply secure, and it is sufficient to check the property at the points where the utility functions are discontinuous. 
In this case, the discontinuities occur when there exists an item $j$ such that all the players allocate zero towards that item. That is, the utility functions are discontinuous at the points in the set 
\[
\mathcal{D} = \{ x \in S | \exists j \in M \; \mbox{such that} \; x_{i,j}=0, \forall i \in N\}
\]
Let $(x^{*}, u^{*})$ be in the closure of the graph of the vector payoff function, where $x^{*} \in \mathcal{D}$. Then $u^{*} = \lim_{K \to \infty} (u_1(x^{K}),$ $\ldots$, $u_n(x^{K}))$ for some $x^K \to x^{*}$.
Let $J$ be the set of items on which no player allocates any weight in $x^{*}$:
\[
J = \{j \in M | x_{i,j}^{*} = 0, \forall i \in N\}
\]

Then there exists a player $i$, an item $l$ and $N_0 \in \mathbb{N}$ such that $\frac{x_{i,l}^{K}}{X_{l}^{K}} \leq \frac{1}{n} + \frac{1}{n^2}$, for all $K \geq N_0$. That is, player $i$ gets item $l$ with probability 
less than $\frac{1}{n} + \frac{1}{n^2}$ for all large enough $K$.

For every item $k \not \in J$, we have: $\lim_{K \to \infty} \frac{x_{i,k}^{K}}{X_{k}^{K}} = \frac{x_{i,k}^{*}}{X_{k}^{*}}$. Then $u_i^{*}$ can be rewritten as:
\[
u_{i}^{*} =\left(\sum_{j \in J} \lim_{K \to \infty} \left(\frac{x_{i,j}^{K}}{X_{j}^{K}} \right) v_{i,j} \right) + \left( \sum_{j \not \in J}  \left( \frac{x_{i,j}^{*}}{X_{j}^{*}} \right) v_{i,j} \right)
\]

Let $\delta > 0$ be small enough, and denote by $L_i$ the set of items on which player $i$ allocates strictly positive weight.
That is, $L_i = \{j \in M | x_{i,j}^{*} > 0 \}$. %Note that $|L_i| > 0$, since $x_{i,j}^{*} = 0$, $\forall j \not \in L_i\}$.
Consider a new strategy profile, $x_{i}^{'}$, for player $i$, such that 
\[
x_{i,j}^{'} = 
\left\{
  \begin{array}{ll}
    \frac{\delta}{|J|} & \mbox{if $j \in J$}\\
    x_{i,j}^{*} - \frac{\delta}{|L_i|} & \mbox{if $j \in L_i$} \\
    x_{i,j}^{*} (=0) & \mbox{otherwise}
  \end{array}
\right.
\]
%
%$x_{i,j}^{'} = \frac{\delta}{|J|}$, $\forall j \in J$, and $x_{i,k}^{'} = x_{i,k}^{*} - \frac{\delta}{m - |J|}$, $\forall k \not \in J$.
The utility of $i$ when playing $x_{i}^{'}$ is:
\[
u_{i}(x_{i}^{'}, x_{-i}^{*}) = \left( \sum_{j \in J} v_{i,j} \right) + \left( \sum_{j \in L_i} \left( \frac{x_{i,j}^{*} - \frac{\delta}{|L_i|}}{X_{j}^{*} - \frac{\delta}{|L_i|}} \right) v_{i,j} \right) 
\]
Let $\delta > 0$ be such that $\delta < \min\left(x_{i,j}^{*} | j \in L_i\right)$ and $\delta < \frac{\left(1 - \frac{1}{n} - \frac{1}{n^2}\right) \cdot |L_i| \cdot v_{i,l}}{\left(\sum_{j \in L_i} \frac{v_{i,j}}{X_{j}^{*}}\right)}$.
We have:
\[
u_i(x_{i}^{'}, x_{-i}^{*}) - u_{i}^{*} \geq \left(1 - \frac{1}{n} - \frac{1}{n^2} \right) v_{i,l} - \left( \frac{\delta}{|L_i|} \right) \left( \sum_{j \in L_i} \frac{v_{i,j}}{X_{j}^{*}} \right) > 0,
\]
%Let $\delta < \min(x_{i,j}^{*} | j \in L_i)$ and $\delta < \frac{\left(1 - \frac{1}{n} - \frac{1}{n^2}\right)|L_i|v_{i,l}}{\sum_{k \not \in J} \frac{v_{i,k}}{X_{k}^{*}}}$.
%
% \in \left(0, \frac{(1-c)(m - |J|) v_{i,l}}{\left( \sum_{k \not \in J} \frac{v_{i,k}}{X_{k}^{*}} \right)}\right)$. 
and so $u_i(x_{i}^{'}, x_{-i}^{*}) > u_{i}^{*}$.
The utility functions are continuous at $x^{'} = (x_{i}^{'}, x_{-i}^{*})$, and so there exists $\varepsilon > 0$ such that $u_i(x_{i}^{'}, y_{-i}) > u_{i}^{*}$ for all $y_{-i} \in B(x_{-i}^{*};\varepsilon)$.
This completes the proof that the game is better-reply secure. Thus the conditions of Theorem \ref{thm:Reny} are met and the game has a pure Nash equilibrium.
\end{proof}

We note that the same result holds (with a very similar proof) when the items that no player placed a bid on are given uniformly at random to a player, rather than being kept by the auctioneer.

\subsection{Discrete Budgets} \label{sec:given_discrete}
In this  section we study the game where the budgets are discrete -- in this case, each player has a number of indivisible tickets.
We refer to the subcase in which each player has exactly one ticket as a game with indivisible budgets.

\begin{theorem}
Chinese auctions with symmetric players and indivisible budgets have a pure Nash equilibrium. 
\end{theorem}
\begin{proof}
Consider the following assignment of players to items:
\begin{itemize}
\item For each player $k \in N$ in decreasing order of ticket weight:
\begin{itemize}
\item Assign the ticket of player $k$ to the item $j$ such that $\left(\frac{w_k}{X_j + w_k}\right)v_j$ is maximal,
where $X_j$ is the weight of the existing tickets at item $j$.
\item $X_j \leftarrow X_j + w_k$
\end{itemize}
\end{itemize}
Assume by contradiction that the assignment is not stable. Then there exists a player $k$ who can deviate, 
by moving his ticket from bin $s_i$ to bin $s_j$, for some $i,j$. For the deviation to be an improvement, 
it must be the case that:
\begin{equation} \label{eq:improv1}
\left(\frac{w_k}{X_j + w_k}\right)v_j > \left(\frac{w_k}{X_i}\right)v_i
\end{equation}
Consider the last player, $l$, who placed a ticket in bin $s_i$. At the time player $l$ placed his ticket, it
must have been the case that bin $s_i$ was preferable to bin $s_j$, i.e. 
\begin{equation} \label{eq:improv2}
\left(\frac{w_l}{X_i}\right)v_i > \left(\frac{w_l}{X_j^{'} + w_l}\right)v_j
\end{equation}
where $X_j^{'}$ was the weight of bin $s_j$ when player $l$ placed his ticket. Since $X_j$ is the final weight 
at item $j$, we have that $X_j \geq X_j^{'}$.
Finally, since the players are assigned in decreasing order of weights, we have that $w_k \geq w_l$, which
combined with equations \ref{eq:improv1} and \ref{eq:improv2} give:
\[
\left(\frac{v_j}{v_i}\right)X_i - X_j > w_k \geq w_l \geq \left(\frac{v_j}{v_i}\right)X_i - X_j
\]
This is a contradiction, thus the assumption must have been false, and the assignment is stable.
\end{proof}

\begin{proposition}
Chinese auctions with symmetric valuations and asymmetric discrete budgets do not necessarily have a pure 
Nash equilibrium, even in the case of two items.
\end{proposition}
\begin{proof}
Consider two players, with budgets $w_1 = 3$ and $w_2 = 1$, respectively, where all the coins have size $1$, and two items, 
such that $v_1 = v_2 = v$.
Consider for example the assignment in which player $1$ assigns two tickets to
item $1$ and one ticket to item $2$, while player $2$ assigns one ticket to item $2$. The expected value 
of player $2$ is $u_2 = \frac{v}{3}$. Player $2$ can deviate by placing his ticket on item $2$ instead, which
would give him higher expected utility: $u_2^{'} = \frac{v}{2} > u_2$. The other assignments can be similarly
verified.
\end{proof}

\begin{proposition}
Chinese auctions with two items, asymmetric valuations, and asymmetric indivisible budgets have a pure Nash equilibrium.
\end{proposition}
\begin{proof}
Consider the following assignment:
\begin{itemize}
\item Assign all the players to item $1$.
\begin{itemize}
\item Iteratively, take the player $i$ with the lowest value of 
$w_i\left(\frac{v_{i,1}}{v_{i,2}}\right)$ among the players bidding on item $1$.
Move player $i$'s ticket to item $2$ if the move improves $i$'s utility.
\end{itemize}
\end{itemize}
The resulting assignment is an equilibrium. 
None of the players at item $1$ have an incentive to move to 
$2$, since the player $i$ who likes $1$ the least (with the lowest ratio $w_i\left(\frac{v_{i,1}}{v_{i,2}}\right)$
among the players at $1$) did not switch. In addition, none of the players at item $2$ have an incentive to
switch back to item $1$, since the last player $j$ who arrived at $2$ does not want to switch, and all the previous 
players at $2$ like this item at least as much as player $j$.
\end{proof}

\begin{theorem}
Chinese auctions with asymmetric valuations and symmetric indivisible budgets have a pure Nash equilibrium.
\end{theorem}
\begin{proof}
It can be verified that the allocation given by Algorithm \ref{alg2} is a pure Nash equilibrium. Note that at each step
during the algorithm, only one player can deviate (at the active item). Moreover, each player can deviate at most once
in each iteration, since the current bin never degrades during the current iteration and the other bins do not improve.
%since once a player deviated, he likes the current bin with the existing number of players. The number
%of players at this bin never gets worse during the current iteration, while the other bins do not improve.
%%%%%%%%%%%%%%%%%%%%SIMINA::TODO::: Write this formally and in more detail%%%%%%%%%%%%%%%%%%%%
\end{proof}

%\vspace{-4mm}
%\restylealgo{ruled}
\begin{algorithm}[htbp]\label{alg2}
%\begin{algorithm}[H]\label{alg1}
%\SetVline
%\linesnumbered
\caption{Equilibrium for Asymmetric Valuations and Symmetric Indivisible Budgets}
  \ForEach{$j \in [m]$} {
    $n_j \leftarrow 0$\\
  }
  \ForEach{player $i \in N$} {
    Assign $i$ to the item $j$ which maximizes $\frac{v_{i,j}}{n_j + 1}$\\
    $n_j \leftarrow n_j + 1$\\
    $a \leftarrow j$ // active item\\
    \While{$\exists$ player $l$ which can deviate from $a$} {
      Move $l$ to the item $k$ which maximizes $\frac{v_{l,k}}{n_k + 1}$\\
      $n_a \leftarrow n_a - 1$\\
      $n_k \leftarrow n_k + 1$\\
      $a \leftarrow k$\\
    }
  }
\end{algorithm}
%\vspace{-2mm}

\section{Costly Tickets} \label{sec:costly}

In this section we analyze the game when the tickets are costly. The costly tickets scenario results in a model similar to the Tullock contest and other rent-seeking problems.
When each player $i$ allocates weight $x_{i,j} \geq 0$ on item $j$, the utility of player $i$ is: 
\begin{equation} \label{def:costly_utility}
u_i(x) = \sum_{j = 1}^{m} \left( \sigma_{i,j}(x) v_{i,j} - x_{i,j} \right)
\end{equation}
When the budgets are discrete, the definition is equivalent to:
\[
u_i(x) = \sum_{j = 1}^{m} \left( \sigma_{i,j}(x) v_{i,j} \right) - w_i
\]

\subsection{Continuous Budgets} \label{sec:costly_continuous}

In this section we analyze the game when the budgets are continuous and the tickets are costly.
First, note that the expected value of player $i$ from an item $j$ is at most $v_j$. Thus in any pure Nash equilibrium, it should be the case that $x_{i,j} \leq v_{i,j}$. Thus it is sufficient 
to study the game when the strategy spaces are restricted to $S_{i}^{m-1} = \{y \in \mathbb{R}^{m} | 0 \leq y_j \leq v_{i,j}, \forall j \in M \}$, for every player $i \in N$.

We have the following results for continuous budgets.

\begin{theorem} \label{thm:symmetric_costly}
Chinese auctions with symmetric valuations and costly continuous budgets have a symmetric pure Nash equilibrium.
\end{theorem}
\begin{proof}
We show that the allocation $x_{i,j} =\left(\frac{n-1}{n^2}\right) v_j$, $\forall i \in N$ and $j \in \{1, \ldots, m\}$ is a pure Nash equilibrium.
For each $i \in N$, the utility of player $i$ when the other players allocate $x$ is:
\begin{eqnarray*}
u_i(y, x_{-i}) & = & \sum_{j = 1}^{m} \left( \frac{y_j \cdot v_j}{ \left( \sum_{k \neq i} \left( \frac{n-1}{n^2} \right) v_j\right) + y_j} - y_j \right) \\
& = & \left( \sum_{j=1}^{m} \frac{y_j \cdot v_j}{\left( \frac{n-1}{n}\right)^{2} v_j + y_j} - y_j \right)
\end{eqnarray*}
Player $i$'s utility can be rewritten as:
\[
u_i(y, x_{-i}) =  \left( \sum_{j=1}^{m} v_j \right) - \left( \sum_{j=1}^{m} \frac{b_j}{a_j + y_j} + y_j \right)
\]
where $a_j = \left( \frac{n-1}{n}\right)^{2} v_j$ and $b_j = \left( \frac{n-1}{n}\right)^{2}v_j^2$, $\forall j \in \{1, \ldots, m\}$.

Let $f:S_{i}^{m-1} \rightarrow \mathbb{R}$, $f(y) = \sum_{j=1}^{m} \left( \frac{b_j}{a_j + y_j} + y_j \right)$. Similarly to Theorem \ref{thm:symmetric_cont_val}, 
$f$ is strictly convex and $y$ maximizes $u_i(y, x_{-i})$ if and only if $y$ is a local minimum of $f$.
It can be verified that $y_j =  \left(\frac{n-1}{n^2}\right) v_j$  is a local minimum of $f$, and so the best response of player $i$ when the other players
allocate $x_{-i}$ is to allocate $x_i$. Thus the game has a symmetric pure Nash equilibrium where each player $i$ allocates $x_{i,j} = \left(\frac{n-1}{n^2}\right) v_j$ 
on item $j$.
\end{proof}

Similarly to the given budgets analysis, the game is not guaranteed to have a pure Nash equilibrium when the valuations can be zero, but the auctioneer
can guarantee the existence of an equilibrium by placing his own ticket in each basket.

\begin{proposition}
Chinese auctions with asymmetric valuations and costly continuous budgets do not necessarily have a pure Nash equilibrium.
\end{proposition}
\begin{proof}
Consider two players with valuations $v_{1,1} = 1$, $v_{1,2} = 0$ and $v_{2,1}=1$, $v_{2,2}=1$.
Assume by contradiction that the game has a pure Nash equilibrium at $(x_1,x_2)$. From Theorem \ref{thm:symmetric_costly}, we have that
$0 \leq x_{i,j} \leq v_{i,j}$, $\forall i,j \in \{1,2\}$.
If $x_{2,2} > 0$, then player $2$ can improve his utility by deviating to $x_2^{'} = (x_{2,1}, \frac{x_{2,2}}{2})$. If $x_{2,2} = 0$, then there
exists $\varepsilon > 0$ such that by playing $x_{2}^{''} = (x_{2,1}, \varepsilon)$, player $2$ gets $u_{2}(x_1, x_2^{''}) > u_{2}(x_1, x_2)$. Thus the game has no
pure Nash equilibrium.
\end{proof}

\begin{theorem} \label{thm:auctioneer_ticket_given}
Chinese auctions with asymmetric valuations and continuous costly budgets have a pure Nash equilibrium when the auctioneer places a strictly positive ticket in each basket.
\end{theorem}
\begin{proof}
Let $\Delta_j > 0$ be the ticket placed by the auctioneer in each basket. For each strategy vector $x$, the utility of player $i$ is:
\begin{equation} \label{eq:utility_delta_ticket_given}
u_i(x) = \sum_{j=1}^{m} \left( \left( \frac{x_{i,j}}{\Delta_j + X_j} \right) v_{i,j} - x_{i,j} \right)
\end{equation}
where $X_{j} = \sum_{k \in N} x_{k,j}$ is the weight placed by all players on item $j$.
The utility function of player $i$, $u_i(x)$ is quasi-concave in $x_i$. The strategy spaces $S_i = \{y \in \mathbb{R}^{m} | 0 \leq y_j \leq v_{i,j}, \forall j \in \{1,\ldots, m\} \}$ are nonempty, compact, and convex.
Moreover, $u_i(x)$ is continuous in $x$ since the denominator of each term in the sum of Equation (\ref{eq:utility_delta_ticket_given}) is strictly positive.
Thus the conditions of Theorem \ref{thm:Glicksberg} apply, and the game has a pure strategy Nash equilibrium when the auctioneer places a ticket in each basket.
\end{proof}

Finally, when the valuations are strictly positive, a pure Nash equilibrium is guaranteed to exist.

\begin{theorem} \label{thm:asymmetric_continuous_given_existence}
Chinese auctions with asymmetric, strictly positive valuations and costly continuous budgets have a pure Nash equilibrium.
\end{theorem}
\begin{proof}
The proof is similar to that of Theorem \ref{thm:asymmetric_continuous_existence}. The strategy spaces $S_i$ are non-empty, compact,
and convex. The utility function of each player $i$ is quasi-concave in $x_i$. The discontinuities occur at the points in the set
\[
\mathcal{D} = \{ x \in S | \exists j \in \{1, \ldots, m\} \; \mbox{such that} \; x_{i,j}=0, \forall i \in N\}
\]
The proof is very similar to that of Theorem \ref{thm:asymmetric_continuous_existence}. That is, 
it can be verified that for any $(x^{*}, u^{*})$ in the closure of the graph of the vector payoff function, where $x^{*} \in \mathcal{D}$, 
there exists a player $i$ and a strategy $x_{i}^{'}$ such that 
$u_i(x_{i}^{'}, y_{-i}) > u_{i}^{*}$ for all $y_{-i} \in B(x_{-i}^{*};\varepsilon)$. The conditions of Theorem \ref{thm:Reny} are met and the game has a pure Nash equilibrium.
\end{proof}

\section{Acknowledgements}
We thank S\o ren Frederiksen and Alejandro Lopez-Ortiz for useful discussion. This work was supported by the Sino-Danish Center for the Theory of Interactive Computation, funded by the Danish National
Research Foundation and the National Science Foundation of China (under the grant 61061130540). The authors also acknowledge
support from the Center for research in the Foundations of Electronic Markets (CFEM), supported by the Danish Strategic Research Council.

\appendix

\section{Appendix}

\setcounter{lemma}{0}
\begin{lemma} \label{lem:convex}
Let $f : S^{m-1} \rightarrow \mathbb{R}$, where 
$S^{m-1} = \{ y \in \mathbb{R}^{m} | y_i \geq 0 \; \mbox{and} \; y_1 + \ldots + y_m = W\}$. Define 
% be such that for all $y = (y_1, \ldots, y_m) \in \Delta^{m-1}$, 
$f(y) = \sum_{j=1}^{m} \frac{b_j}{a_j + y_j}$, where $a_j, b_j > 0, \forall j \in \{1, \ldots, m\}$.
Then $f$ is strictly convex.
\end{lemma}
\begin{proof}
The domain of $f$ is convex, thus it is sufficient to verify that for all $y \neq z \in S^{m-1}$ and $\lambda \in (0,1)$, $f(\lambda y + (1 - \lambda) z) < \lambda f(y) + (1-\lambda) f(z)$.
For each $j \in \{1, \ldots, m\}$, let $f_j: \mathbb{R}^{+} \leftarrow \mathbb{R}$, $f_j(x) = \frac{b_j}{a_j + y_j}$. Then $f_{j}^{''} = \frac{2 b_j}{(a_j + y_j)^{3}} > 0$, and $f_j$ is strictly convex. Thus
$\frac{a_j \cdot b_j}{a_j + \lambda y_j + (1 - \lambda) z_j} \leq \lambda \left(\frac{a_j \cdot b_j}{a_j + y_j}\right) + (1 - \lambda) \left(\frac{a_j \cdot b_j}{a_j + z_j}\right)$ $(^*)$
with equality if and only if $y_j = z_j$. By summing $(^*)$ over all $j$ and noting that at least one inequality is strict (since $y \neq z$), we obtain that $f$ is strictly convex.
\end{proof}

\begin{comment}
\begin{lemma} \label{lem:quasi_concave}
Let $f : S^{m-1} \rightarrow \mathbb{R}$, where 
$S^{m-1} = \{ y \in \mathbb{R}^{m} | y_i \geq 0 \; \mbox{and} \; y_1 + \ldots + y_m = W\}$.
For each $j \in \{1, \ldots, m\}$, let $a_j, b_j \geq 0$ and define
$f_j : \mathbb{R}^{+} \rightarrow \mathbb{R}$, 
$f_j(x) =  \frac{b_j \cdot x}{a_j + x}$ if $x > 0$ and $f_j(x) = 0$ otherwise.
Let $f(y) = \sum_{j=1}^{m} f_j(y_j)$. Then $f$ is quasi-concave.
\end{lemma}
\begin{proof}
For each $y \neq z \in S^{m}$, the inequality
\begin{equation} \label{eq:quasi_concave_f_j}
f_j(\lambda y_j + (1 - \lambda) z_j) \geq \min(f_j(y_j), f_j(z_j))
\end{equation}
holds if and only if
\begin{itemize}
\item $f_j(0) \geq f_j(0)$ when $y_j = z_j = 0$
\item $f_j(\lambda y_j) \geq \min(f_j(y_j), f(0)) = 0$ when $y_j > 0$ and $z_j = 0$
\item $f_j(\lambda y_j + (1 - \lambda) z_j) \geq \min(f_j(y_j), f_j(z_j))$ when $y_j, z_j > 0$. It can be easily verified that
\begin{small}
\begin{eqnarray*}
f_j(\lambda y_j + (1 - \lambda) z_j) &=& \frac{b_j (\lambda y_j + (1- \lambda) z_j)}{a_j + \lambda y_j + (1 - \lambda) z_j} \geq 
\lambda \left(\frac{b_j \cdot y_j}{a_j + y_j}\right) + (1 - \lambda) \left(\frac{b_j \cdot z_j}{a_j + z_j}\right)  \\
&=& \lambda f_j(y_j) + (1 - \lambda) f_j(z_j) \geq \min(f_j(y_j), f_j(z_j))
\end{eqnarray*}
\end{small}
\end{itemize}
Thus Inequality (\ref{eq:quasi_concave_f_j}) holds. By summing over all $j \in \{1, \ldots, m\}$, it follows that $f(\lambda y + (1 - \lambda)z) \geq \min(f(y), f(z))$, and so $f$ is quasi-concave.
\end{proof}
\end{comment}


\begin{thebibliography}{}
%\nocite{*}
\bibitem{Babaioff09}
M.~Babaioff, M.~Dinitz, A.~Gupta, N.~Immorlica, and K.~Talwar.
\newblock Secretary problems: weights and discounts.
\newblock In {\em SODA}, pages 1245--1254, 2009.

\bibitem{Babaioff08}
M.~Babaioff, N.~Immorlica, D.~Kempe, and R.~Kleinberg.
\newblock {Online Auctions and Generalized Secretary Problems}.
\newblock In {\em ACM SIGecom Exchanges}, volume~7, 2008.

\bibitem{Bachrach03}
Y.~Bachrach and J.~S. Rosenschein.
\newblock {Coalitional Skill Games}.
\newblock In {\em AAMAS}, pages 1023--1030, 2003.

\bibitem{Chalkiadakis10}
G.~Chalkiadakis, E.~Elkind, E.~Markakis, M.~Polukarov, and N.~R. Jennings.
\newblock {Cooperative Games with Overlapping Coalitions}.
\newblock {\em Journal of Artificial Intelligence Research (JAIR)},
  39:179--216, 2010.

\bibitem{Chowdhury11}
S.~M. Chowdhury and R.~M. Sheremeta.
\newblock {A generalized Tullock contest}.
\newblock {\em Public Choice}, 147(3):413--420, 2011.

\bibitem{Clark98}
B.~Clark and C.~Riis.
\newblock Competition over more than one prize.
\newblock {\em The American Economic Review}, 88(1):276--289, 1998.

\bibitem{Palma12}
A.~de~Palma and S.~Munshi.
\newblock {Multi-player, Multi-prize, Imperfectly Discriminating Contests}.
\newblock {\em Working Paper}, 2012.

\bibitem{Fang02}
H.~Fang.
\newblock Lottery versus all-pay auction models of lobbying.
\newblock {\em Public Choice}, 112(3-4):351--71, 2002.

\bibitem{Fatima06}
S.~Fatima.
\newblock Sequential versus simultaneous auctions: a case study.
\newblock In {\em ICEC}, 2006.

\bibitem{Gradstein89}
M.~Gradstein and S.~Nitzan.
\newblock Advantageous multiple rent seeking.
\newblock {\em Mathematical and Computer Modelling}, 12(4-5):511--518, 1989.

\bibitem{Hillman89}
A.~Hillman and J.~Riley.
\newblock Politically contestable rents and transfers.
\newblock {\em Economics and Politics}, (1):17--39, 1989.

\bibitem{Kleinberg05}
R.~Kleinberg.
\newblock A multiple-choice secretary algorithm with applications to online
  auctions.
\newblock In {\em SODA}, pages 630--631, 2005.

\bibitem{Matros07}
A.~Matros.
\newblock Chinese auctions.
\newblock {\em Mimeo, University of Pittsburgh}, 2007.
\newblock http://www.gtcenter.org/Archive/Conf07/Downloads/Conf/Matros478.pdf.

\bibitem{Moldovanu01}
B.~Moldovanu and A.~Sela.
\newblock The optimal allocation of prizes in contests.
\newblock {\em American Economic Review}, 91(3):542--558, 2001.

\bibitem{AGT}
N.~Nisan, T.~Roughgarden, E.~Tardos, and V.~Vazirani, editors.
\newblock {\em {Algorithmic Game Theory}}.
\newblock Cambridge University Press, 2007.

\bibitem{Nisan11}
N.~Nisan, M.~Schapira, G.~Valiant, and A.~Zohar.
\newblock {Best-Response Auctions}.
\newblock In {\em ACM EC}, 2011.

\bibitem{Nitzan94}
S.~Nitzan.
\newblock Modelling rent-seeking contests.
\newblock {\em European Journal of Political Economy}, 10(1):41--60, 1994.

\bibitem{Nti99}
K.~O. Nti.
\newblock Rent seeking with asymmetric valuations.
\newblock {\em Public Choice}, 98(3-4):415--430, 1999.

\bibitem{Reny99}
P.~J. Reny.
\newblock {On the Existence of Pure and Mixed Strategy Nash Equilibria in
  Discontinuous Games}.
\newblock {\em Econometrica}, 67(5):1029--1056.

\bibitem{Reny06}
P.~J. Reny.
\newblock {Non-Cooperative Games: Equilibrium Existence, in S. N. Durlauf and
  L. E. Blume: The New Palgrave Dictionary of Economics. Basingstoke: Palgrave
  Macmillan}.
\newblock 2008.

\bibitem{Blotto}
B.~Robertson.
\newblock {The Colonel Blotto game}.
\newblock {\em Economic Theory}, 29(1):1--24, 2006.

\bibitem{MAS}
Y.~Shoham and K.~Leyton-Brown.
\newblock {\em {Multiagent Systems - Algorithmic, Game-Theoretic, and Logical
  Foundations}}.
\newblock Cambridge University Press, 2009.

\bibitem{Siegel09}
R.~Siegel.
\newblock {All-Pay Contests}.
\newblock {\em Econometrica}, 77(1):71--92, 2009.

\bibitem{Skaperdas96}
S.~Skaperdas.
\newblock Contest success functions.
\newblock {\em Economic Theory}, 7(2):283--290, 1996.

\bibitem{Tullock80}
G.~Tullock.
\newblock {Efficient rent-seeking, in: J.M. Buchanan, R.D. Tollison, G. Tullock
  (Eds.), Toward a theory of the rent-seeking society (Texas A. \& M.
  University Press, College Station, TX)}.
\newblock pages 97--112, 1980.

\bibitem{wiki}
Wikipedia.
\newblock {Chinese Auction}.
\newblock http://en.wikipedia.org/wiki/Chinese\textunderscore auction.


\end{thebibliography}
\end{document}